\immediate\write18{makeindex FINALVERSION.nlo -s nomencl.ist -o FINALVERSION.nls}
\documentclass[journal]{IEEEtran}

\usepackage{cite}

\ifCLASSINFOpdf
   \usepackage[pdftex]{graphicx}
\else
\fi
\usepackage{amsmath}
\usepackage{amsthm}
\usepackage{amssymb,amsfonts} 
\usepackage{mathtools} 
\usepackage{algorithmic}

\usepackage{array}

\usepackage[table]{xcolor}

\ifCLASSOPTIONcompsoc
  \usepackage[caption=false,font=normalsize,labelfont=sf,textfont=sf]{subfig}
\else
  \usepackage[caption=false,font=footnotesize]{subfig}
\fi

\usepackage{nomencl}
\makenomenclature

\usepackage{dsfont} 
\usepackage{multirow} 

\newtheorem{prop}{Proposition}

\begin{document}
%
\title{A Graph-Theoretic Approach for Spatial Filtering and Its Impact on Mixed-type Spatial Pattern Recognition in Wafer Bin Maps}
%
%
%

\author{Ahmed Aziz Ezzat,~\IEEEmembership{Member,~IEEE,}
        Sheng Liu, Dorit Hochbaum, and Yu Ding,~\IEEEmembership{Senior Member,~IEEE}
\thanks{A. A. Ezzat is with the Department
of Industrial \& Systems Engineering Engineering, Rutgers University, E-mail: aziz.ezzat@rutgers.edu}
\thanks{S. Liu is with the Rotman School of Management, University of Toronto}
\thanks{D. Hochbaum is with the Department of Industrial Engineering \& Operations Research, University of California, Berkeley}
\thanks{Y. Ding is with the Department of Industrial \& Systems Engineering, Texas A\&M University}
}
%
%

\markboth{\copyright[$2021$] IEEE. This work has been accepted at the IEEE for publication.}%
{Aziz Ezzat \MakeLowercase{\textit{et al.}}: Spatial Pattern Recognition via Adjacency-Clustering}
%



\maketitle

\begin{abstract}
Statistical quality control in semiconductor manufacturing hinges on effective diagnostics of wafer bin maps, wherein a key challenge is to detect how defective chips tend to spatially cluster on a wafer\textemdash a problem known as \textit{spatial pattern recognition}. Recently, there has been a growing interest in \textit{mixed-type} spatial pattern recognition\textemdash when multiple defect patterns, of different shapes, co-exist on the same wafer. Mixed-type spatial pattern recognition entails two central tasks: (1) spatial filtering, to distinguish systematic patterns from random noises; and (2) spatial clustering, to group filtered patterns into distinct defect types. Observing that spatial filtering is instrumental to high-quality mixed-type pattern recognition, we propose to use a graph-theoretic method, called adjacency-clustering, which leverages spatial dependence among adjacent defective chips to effectively filter the raw wafer maps. Tested on real-world data and compared against a state-of-the-art approach, our proposed method achieves at least $46$\% gain in terms of internal cluster validation quality (i.e., validation without external class labels), and about $\hspace{-0.5mm}5$\% gain in terms of Normalized Mutual Information\textemdash an external cluster validation metric based on external class labels. Interestingly, the margin of improvement appears to be a function of the pattern complexity, with larger gains achieved for more complex-shaped patterns. 
\end{abstract}

\begin{IEEEkeywords}
Clustering, Graph theory, Pattern recognition, Spatial data science, Unsupervised learning, Wafer bin maps. 
\end{IEEEkeywords}

\IEEEpeerreviewmaketitle

\section{Introduction}

\IEEEPARstart{I}{ntegrated} circuits (IC), colloquially known as chips, are essential to most, if not all, electronic devices. The central step in IC manufacturing is wafer fabrication, in which a batch of chips are fabricated on round-shaped silicon wafers through a series of complex electrochemical processes including slicing silicon-rich ingots into round-shaped thin wafers, wafer oxidation and material deposition, photolithography, ion implantation, and etching \cite{el2003reusable}. Once fabricated, all wafers undergo an operational quality performance test, known as wafer probing, in which chips are labeled as functional or defective by passing an input signal and collecting the corresponding output. This step results in a two-dimensional graphical representation called a \textit{wafer bin map}\textemdash a gridded representation of a wafer in which each grid point represents the spatial location of a chip and is assigned a binary value (e.g., $0$ or $1$) denoting a functional or a defective chip, respectively. Figure \ref{fig:mot1} shows examples of three different wafer bin maps resulting from wafer probing tests, where defective chips are denoted by red squares. \par

\begin{figure}[h]
    \centering
    \includegraphics[trim = 0 0.5cm 0 0.75cm, width = 1\linewidth]{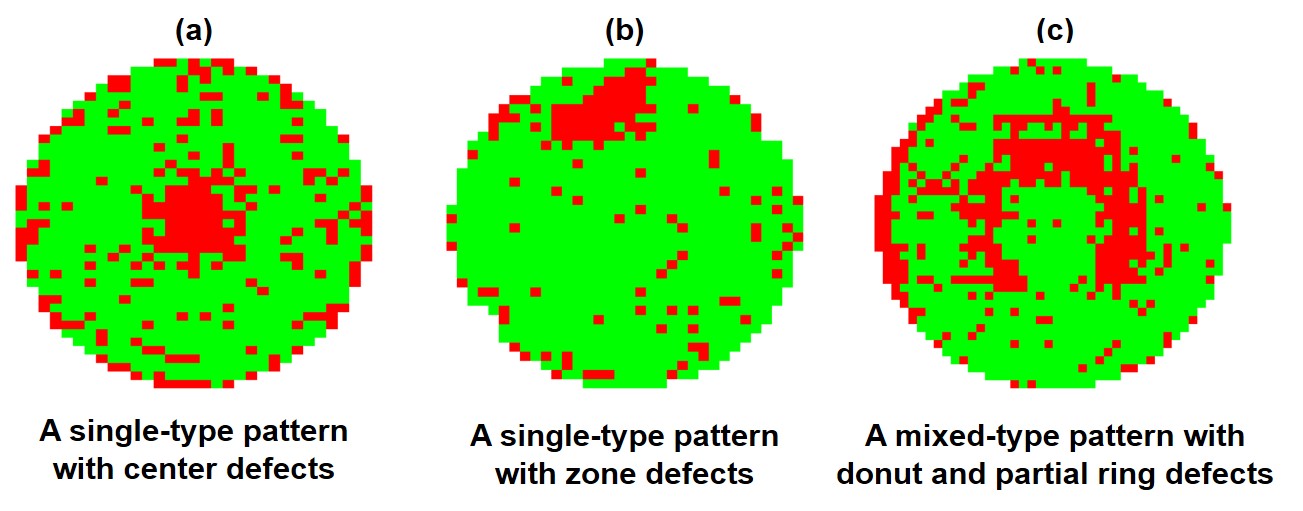}
    \caption{Examples of wafer bin maps. Panels (a-b) are \textit{single-type} patterns, i.e. one defect pattern per wafer, while Panel (c) shows a \textit{mixed-type} defect pattern. \textit{Random patterns} due to inherent process variation are represented by the scattered red squares on the wafer maps, which overlap with \textit{systematic patterns} (e.g. center, zone, donut) that are attributed to assignable root causes.
    }
    \label{fig:mot1}
\end{figure}

A careful analysis of wafer bin maps is pivotal to quality control efforts in the semiconductor manufacturing industry. By investigating the spatial defect patterns on the fabricated wafers, i.e., how the defective chips tend to spatially cluster, one can infer instrumental insights about the root causes of defect occurrence, and subsequently suggest corrective actions to mitigate future failures. This problem, often referred to in the literature as \textit{spatial pattern recognition} (SPR), is the focus of this paper. SPR 
is of extreme importance to pinpoint possible root causes of failures in wafer fabrication. In fact, several spatial defect patterns in wafer bin maps can be directly traced to common root causes of failure. For instance, a circular-shaped, center-located defect pattern, as shown in Figure \ref{fig:mot1}(a), often corresponds to chemical stains or mechanical equipment faults \cite{hwang2007model,kim2018detection}, while a zone-shaped, edge-located defect pattern, as shown in Figure \ref{fig:mot1}(b), can be traced to uneven polishing or edge-die effects \cite{cunningham1998statistical}. A center-located, donut-shaped defect pattern, as shown in Figure \ref{fig:mot1}(c), is routinely observed in wafer data due to possible tooling problems \cite{neo2017failure}. These spatial patterns like center, zone, or donut, wherein defective chips ``cluster'' to form distinct shapes are referred to as \textit{systematic patterns} since they correspond to an assignable root cause. In contrast, randomly scattered defective chips in Figure \ref{fig:mot1}(a)-(c) are called \textit{random patterns}, or \textit{noises}, since they are merely artifacts of random process variation. \par 

In the SPR literature, defect patterns in wafer bin maps can either be \textit{single-type} or \textit{mixed-type}. The former refers to wafer maps that host only one defect pattern, while the latter refers to wafer maps in which two or more defect patterns co-exist. Figure \ref{fig:mot1}(a-b) show examples of single-type defect patterns, whereas Figure \ref{fig:mot1}(c) depicts a mixed-type defect pattern. With the ever-growing increase in scale and sophistication of the wafer fabrication process, mixed-type patterns are increasingly observed in production data. Barring few recent efforts \cite{kim2018detection,tello2018deep,kyeong2018classification, kong2019recognition, lee2020semi, wang2020deformable}, the problem of mixed-type SPR has received less attention relative to its single-type counterpart. 

A typical SPR analysis of a wafer bin map, be it hosting single- or mixed-type patterns, involves two pillar tasks: (1) Spatial filtering, i.e., to de-noise raw wafer data by separating systematic from random patterns; and (2) Spatial clustering, i.e., to group the filtered patterns into one or more sub-clusters pertaining to different types of defect patterns (e.g., center, donut). The overwhelming majority of the literature has been devoted to improving the effectiveness of the second task, namely spatial clustering, among which those that are based on model- or density-based clustering \cite{yuan2007model,yuan2008spatial,hwang2007model, yuan2011detection, jin2019novel}, kernel-based clustering \cite{wang2009separation, chao2009wafer}, similarity-based metrics such as correlograms, nearest-neighbor measures, or Voronoi-based partitioning \cite{taam1993detecting, jeong2008automatic, taha2017clustering}, feature extraction-based approaches such as those utilizing Hugh transforms, single value decomposition, or mask-based features \cite{white2008classification, liu2014detecting, kim2015regularized, wang2019wafer}, decision trees and manifold learning \cite{piao2018decision, yu2015wafer}, regression-based consensus networks and ensemble learning \cite{piao2018decision,adly2015randomized, saqlain2019voting}, and neural networks, especially those based on adaptive resonance theory \cite{chen2000neural, su2002neural, di2005unsupervised, huang2007clustered,liu2013intelligent, chien2013system}, or on deep learning-based architectures \cite{lee2016deep, nakazawa2018wafer, yu2019wafer, imoto2019cnn, kong2020semi, hwang2020variational, tsai2020light, jang2020support, hyun2020memory}. \par

On the other hand, methods for the first task, namely spatial filtering, are mostly dominated by \textit{ad hoc} heuristics intended to pre-process or denoise the raw wafer data, with an implicit assumption that the deficiencies of a poorly designed filtering step will be ultimately corrected in the second task (spatial clustering). While this assumption may be acceptable for single-type SPR, we claim that spatial filtering is of extreme importance to mixed-type SPR. Motivated by a similar observation, an algorithm called \textit{connected path filtering} (CPF) has been recently proposed to filter mixed-type wafer bin maps \cite{kim2018detection}. The authors proposed to pair CPF with a spatial clustering model that acts on the filtered data to produce the final SPR results. CPF is a heuristic algorithm that searches all possible connected paths of defective chips on a wafer and only keeps those paths that are longer than a pre-set threshold, $M$. \par 

While valuable on its own, our analysis of multiple wafer maps, as we will elaborate in the sequel, has revealed two main limitations of the CPF approach. First, CPF does not directly leverage local spatial neighborhood information, but instead, it disregards all defective chips that do not belong to a connected path which is longer than a globally pre-set value, $M$. In other words, if a chip is labeled as ``functional,'' while all of its neighbors are not, CPF does not make use of the local neighborhood information to possibly re-assign the label of this chip. A direct consequence of this limitation is that the filtered results may end up having irregular shapes (few functional chips surrounded by large groups of defective chips, or vice-versa), which may severely compromise the quality of the downstream clustering task. To demonstrate this first limitation, let us take a look at Figure \ref{fig:motiv1}, where Panel (a) shows a raw wafer bin map with a mixed-type pattern comprising partial ring and donut defects. The results from CPF (Panel b) clearly show an irregular shape due to either functional chips for which the values should have been updated to defective based on their local neighborhood, or alternatively, defective chips which should have been re-labeled as functional. This irregularity in the shapes misleads the downstream spatial clustering (performed using a mixture model\textemdash to be reviewed later) to mistakenly identify some defective chips as independent sub-clusters. The results in Panel (c) appear to be more visually appealing with a clear visual distinction between the two overlapping defect types, making the downstream clustering method (performed using the same mixture model) relatively straightforward. The result in Panel (c) is in fact, produced by our proposed filtering approach, for which the details are unraveled in Section \ref{sec:approach}. \par

\begin{figure} [h]
\centering
    \includegraphics[trim = 0 0.75cm 0 0.75cm, width = 1\linewidth]{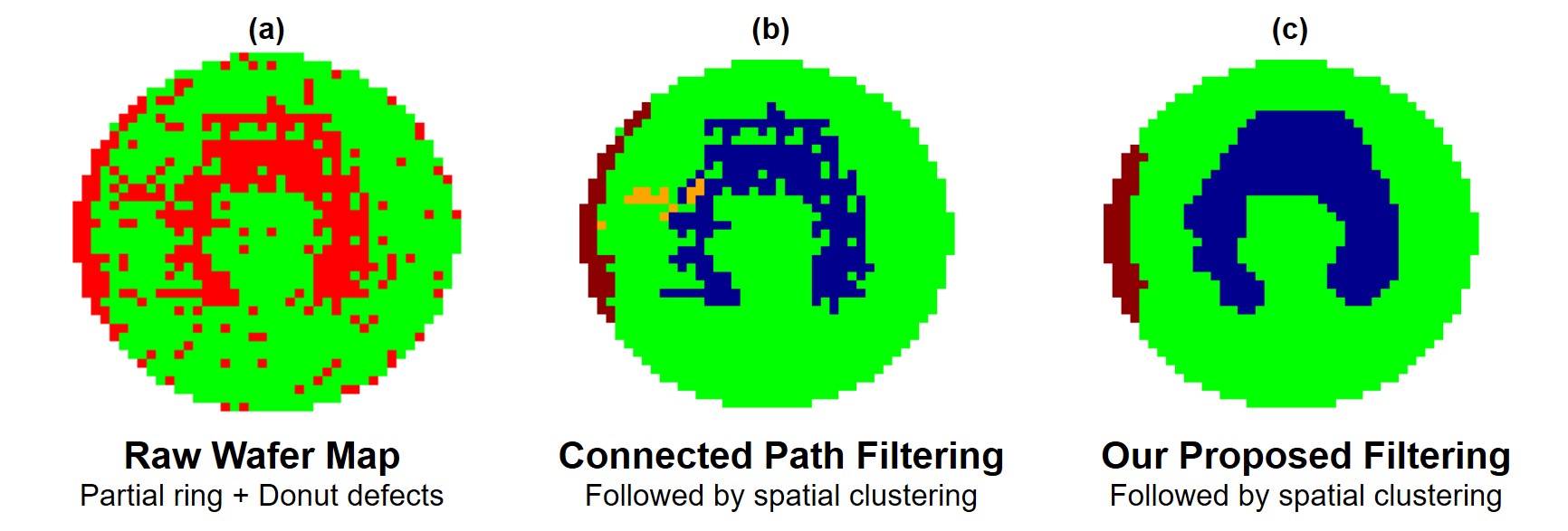}
    \caption{Analysis of a wafer map with overlapped partial ring and donut defects. (a): raw wafer bin map, (b): CPF with $M$ = 5, and (c): Better filtering results produced by our proposed method.}
    \label{fig:motiv1}
\end{figure}

The second limitation of the CPF approach is its choice of $M$. Our findings, to be presented in Section \ref{sec:application}, suggest that there does not appear to be a default value for $M$ that works universally well for all wafers and combinations of defect types, making the choice of $M$ wafer-specific. As outlined by the authors in \cite{kim2018detection}, this choice is made via interaction with domain experts. This limitation may severely hamper the applicability of CPF in practice. Given the large production volumes typical of wafer production lines, the need for domain experts to constantly weigh in and update the value of $M$ can be extraordinarily inefficient and thus impractical. \par 

Motivated by the need to address those two limitations, we propose an innovative approach for mixed-typed SPR in wafer bin maps. Our major finding is that a graph-theoretic approach which leverages the local spatial dependence structure can considerably improve the spatial filtering step, and ultimately, the overall SPR quality. Specifically, we propose to use the adjacency-clustering (AC) method, which was originally introduced by \cite{hochbaum2018adjacency} for yield prediction. Although technically similar, the main function of AC in our work is different from that in \cite{hochbaum2018adjacency}: we focus on extracting systematic defect patterns (i.e., diagnostic analysis), rather than yield prediction (i.e., prognostic analysis). This distinction drives a fundamental departure in what constitutes a cluster: In \cite{hochbaum2018adjacency}, they define a cluster as a group of chips with homogeneous yield level, while our approach defines a cluster based on its chips' membership in either the set of systematic or random defect pattern clusters. AC is closely tied to Markov Random Field (MRF) models which have been successfully applied in image segmentation and restoration \cite{rue2005gaussian, geman1986markov}, spatial clustering \cite{hansen1997monitoring}, and wafer thickness variation analysis \cite{bao2014hierarchical}.

We couple the proposed spatial filtering method with a mixture model for spatial clustering. Based on the numerical experiments with real-world data, we show that our approach outperforms the state-of-the-art method in \cite{kim2018detection} by at least $46$\% in terms of internal cluster validation quality (i.e. validation without external information about class labels), and about $\hspace{-0.5mm}5$\% in terms of Normalized Mutual Information\textemdash an external cluster validation metric which makes use of externally provided class labels. The mixture model used for spatial clustering is the same as that used by \cite{kim2018detection} so that the improvements in the final clustering quality are solely attributed to our proposed filtering method. Interestingly, the margin of improvement appears to be a function of the defect pattern complexity, with larger gains achieved for more complex-shaped patterns. 

In summary, our main contribution is to propose a graph-theoretic spatial filtering method which effectively distinguishes the random noises that are prevalent in wafer bin maps from systematic patterns, which are often attributed to well-studied root causes. There is an overwhelming evidence in the distant and recent literature that a poorly designed filtering method can have substantial detrimental impacts on the SPR performance in wafer bin maps, even with the emergence of powerful deep learning-based approaches \cite{yu2019wafer,kong2020semi,tsai2020light,tello2018deep}. Unlike existing filtering and pre-processing methods, our approach fully leverages the spatial dependence information (via its graph-theoretic structure), and is not vulnerable to arbitrary parameter selections (e.g. filter size or length thresholds) which can be wafer-specific, thus requiring continuous intervention of domain experts. Furthermore, our method is shown to have a desirable combinatorial structure which can be solved in polynomial time by a minimum-cut algorithm. When coupled with a spatial clustering approach (in this paper, a mixture model), substantial improvements in SPR performance are detected, owing to the abovementioned advantages. In principle, we could replace this mixture model by any clustering or classification approach, thus emphasizing the generality and substantial benefit brought about by our proposed spatial filtering method.

We conclude this Section by describing the organization of this paper. In Section \ref{sec:approach}, we elucidate the building blocks of our proposed approach, which comprises the details of the AC approach to filter the wafer maps, coupled with an advanced mixture model to further group the AC-filtered results into one or more sub-clusters corresponding to different systematic defect patterns. Section \ref{sec:application} presents our case study which details the analysis of {twelve} real-world wafer bin maps exhibiting complex multi-type defect patterns. Section \ref{sec:conclus} concludes this paper and highlights future research directions. 

\section{Our Approach}\label{sec:approach}
We represent a wafer map with $n$ chips by $(d_1, d_2, \dots, d_n)$, where $d_i \in \mathbb{N}$ is the number of defects on chip $i$. The locations of chips can be modeled as a graph $G = (V,E)$ where nodes denote chips and the edges define the neighborhood relationship, i.e., we have an edge $[i,j]\in E$ when chip $i$ and chip $j$ are adjacent to each other on the wafer map. According to the neighborhood system, each chip can have at most four neighbors (rook-move neighborhood) or eight neighbors (king-move neighborhood). 

Our SPR framework consists of two stages: a spatial filtering stage, and a spatial clustering stage, both of which are clustering tasks, yet they serve different purposes. In the first stage, namely spatial filtering, AC partitions the wafer map into two clusters, such that one of them only includes those chips that form systematic defect patterns. As a result, we are able to separate the systematic patterns (those caused by assignable root causes) from the random patterns (artifacts of random process variation). In the second stage, the AC-filtered results are further partitioned into one or more sub-clusters using a mixture model called the infinite warped mixture model (iWMM). Each sub-cluster corresponds to a type of systematic defect pattern, e.g., a center or zone, as shown in Figure \ref{fig:mot1}. 

\subsection{Adjacency-Clustering for Spatial Filtering} \label{subsec:ac}
We sketch here the AC model from \cite{hochbaum2018adjacency} and present how it can be adapted into the mixed-type SPR problem. The AC model aims to partition the set of chips into clusters such that chips belonging to the same cluster behave similarly and tend to be adjacent to each other. This AC concept is motivated by the spatial dependence among adjacent chips, which aligns well with the concept of the systematic defect patterns on a wafer map where defective chips tend to spatially cluster. In the case of binary defect data (i.e., $d_i\in\{0,1\}$), the AC model will find two clusters: the first cluster corresponds to the set of systematic defect patterns, while the other corresponds to random defect patterns. 

The clustering decisions are cluster labels $x_i$ for $i\in V$. Chips with the same label form a single cluster. The objective function of AC includes a deviation cost function and a separation cost function. The deviation cost function measures how $x_i$ deviates from the observed value $d_i$, while the separation cost function captures the difference in assigned labels of adjacent chips. Let $f_i(x_i, d_i)$ denote the deviation cost functions associated with node $i\in V$ and $g_{ij}(x_i - x_j)$ denote the separation cost functions for edge $[i,j]\in E$. AC can be formulated as the following integer program:
\begin{align}
\tag{AC}
\min~&\sum_{i\in V}f_i(x_i, d_i)+
\sum_{[i,j]\in E} g_{ij}(x_i-x_j)\\
s.t.~~  &x_i\in X \quad\forall~i\in V,
\end{align}
where $X$ is the set of allowable labels of each chip. In our application, we have $X = \{0,1\}$.

The clustering results depend on the trade-off between the two cost functions. When the separation function values are relatively larger than the deviation function values, the resulting clusters will be more contiguous (the spatial smoothing effect is more significant). If the separation costs are too large, the whole wafer map would be forced to have the same label so the separation cost is minimized. On the other hand, when the separation costs are small, the assigned labels will be close to the original observational values and the spatial filtering effect is less notable in the clustering result. The AC model has a statistics foundation from Markov random fields (MRF) wherein solving for $x_i$ is finding the maximum a posterior (MAP) estimate of the degradation model with an MRF prior \cite{geman1986markov}. Different forms of separation and deviation functions reflect different distributional assumptions of MRF. 

The neighborhood system also plays an important role in the AC results. When the rook-move neighborhood structure is assumed, the clustering only looks for defect patterns that grow horizontally or vertically. By contrast, with the king-move structure, the clustering will identify defect patterns that exhibit more complex shapes such as ring and donut patterns. Therefore, the king-move structure can work better for complicated clustering tasks, as those prevalent in mixed-type defect detection (See Figure \ref{fig:motiv1}(c) for an example).

When both $d_i$ and $x_i$ are binary ($X=\{0,1\}$), as in our SPR application, then the AC model reduces to the problem called minimum s-excess \cite{hochbaum2001efficient}:
\begin{align}
\tag{AC-BIN}
\min \ & \sum_{i\in V}w_i x_i + \sum_{[i,j]\in E} u_{ij}z_{ij},  \\
s.t. \ & z_{ij} \geq x_i - x_j \ \forall \ [i,j] \in E    \label{const_z},\\
 \ & z_{ij} \geq x_j - x_i \ \forall \ [i,j] \in E   ,\\
       & x_i \in \{0,1\} \ \forall\ i\in V,\ z_{ij} \in \{0,1\} \ \forall \ [i,j] \in E. \label{binary}
\end{align}
where $z_{ij} = |x_i - x_j|\in\{0,1\}$ indicates the difference in the label values of chip $i$ and $j$, while $w_i$ is the deviation cost of chip $i$, and $u_{ij}$ is the separation cost associated with the pair of chips. More specifically, $w_i=f_i(1, 0) >0 $ for chips with $d_i=0$ and $w_i = -f_i(0,1)<0$ for chips with $d_i = 1$: (1) when $d_i = 0$, we will incur a penalty of $f_i(1,0)$ for labeling $x_i = 1$ and zero penalty otherwise, so the associated deviation cost is $f_i(1,0)\cdot x_i$ and $w_i = f_i(1,0)$; (2) when $d_i = 1$, we will incur a penalty of $f_i(0,1)$ when assigning $x_i =  0$ and zero penalty otherwise, hence the associated deviation cost is $f_i(0,1)\cdot (1 - x_i)$; After dropping the constant, we get $w_i = - f_i(0,1)$. And $u_{ij} = g_{ij}(1)>0$ for all pairs (the separation cost is 0 if $z_{ij}=0$). The reduction to the minimum s-excess is attainable for any deviation and separation functions.

The minimum s-excess problem can be solved in polynomial time with a minimum-cut algorithm applied to an appropriately defined graph \cite{hochbaum2001efficient}:
\begin{prop} \textit{The adjacency-clustering model with binary label values (AC-BIN) can be solved in polynomial time.}
\end{prop}
\begin{proof}
First, we can verify that the constraint matrix of (AC-BIN) is totally unimodular and therefore, (AC-BIN) can be solved in polynomial time. Second, the constraints also correspond to the dual of the minimum cost network flow problem. Then finding the solution to (AC-BIN) is equivalent to finding the minimum cut on a graph adapted from $G$ \cite{hochbaum2001efficient}.
\end{proof}
The algorithm constructs a graph $G_{st} = (V\cup\{s,t\}, A_{st})$ as follows: First we add to $G$ a source node $s$ and a sink node $t$, and each edge $[i,j]$ is replaced by two arcs $(i,j)$ and $(j,i)$ with the same capacity of $u_{ij}$. For each node $i\in V$ with a positive $w_i$, we add an arc of capacity $w_i$ from the node to the sink. For each node $j\in V$ with a negative weight $w_j$, we add an arc of capacity $-w_j$ from the source. Then the defective cluster (the set of chips with $x_i=1$) is the source set of a minimum cut in $G_{st}$. The computational results indicate that this algorithm can solve instances with thousands of chips within seconds, which facilitates its real-time adoption in practice.

If $d_i$ and $x_i$ take more than two values, the AC model can still be solved in polynomial time for convex deviation and separation functions. Specifically, for ``bilinear" separation cost functions (i.e., $g(x_i - x_j) = u_{ij}\cdot (x_i - x_j)$ if $x_i \geq x_j$ and $u_{ji} \cdot (x_j - x_i)$ otherwise) and any convex deviation function, \cite{hochbaum2001efficient} devised an algorithm that solves the problem in the running time of a single minimum cut (and the running time of finding the minima of the convex deviation functions). This time complexity was also shown by \cite{hochbaum2001efficient} to be the best that can be achieved.  When the separation cost function is not ``bilinear" but convex, the Lagrangian relaxation technique can be applied for the polynomial time algorithm \cite{ahuja2003solving} . 

After solving AC with binary label values, each chip is assigned a new label. The chips with a label value of one form a cluster that contains systematic defect patterns, while the chips with a label value of zero are filtered out. In other words, the original defects recorded on the zero-label chips are treated as random defect patterns of nonassignable causes, to be marked off by the spatial filtering stage and thus no longer deemed defects in the subsequent spatial clustering stage. The spatial filtering result depends on the relative magnitude of the separation costs and deviation costs (the relative differences between $w_i$ and $u_{ij}$ in AC-BIN). As we show in Section III, our numerical analysis suggests that there is a set of parameter values that yields consistently high quality SPR results. 

\subsection{Infinite Warping Mixture Model for Spatial Clutering} \label{sec:iwmm}
Given the AC-filtered results, we apply iWMM \cite{iwata2012warped} to group the resulting systematic patterns into sub-clusters pertaining to distinct types of defect patterns. Before we elaborate on the details of iWMM, we first briefly discuss the motivation of using it in our setting. iWMM was first proposed by \cite{iwata2012warped} and then adopted by \cite{kim2018detection} to spatially cluster the wafer maps that are filtered via CPF. In our approach, we keep the iWMM as our spatial clustering method, because iWMM is a highly potent multi-class clustering method and lends itself well to the SPR problem (more about this in the following). Additionally, by using the same spatial clustering method as that used by \cite{kim2018detection}, we ensure that the improvements in SPR quality are mainly attributed to our proposed filtering method. \par

The benefit of using iWMM in spatial clustering of wafer defect patterns is two-fold. First, in iWMM, the number of sub-clusters corresponding to the number of defect types on a wafer is estimated rather than specified \textit{a priori}\textemdash a common shortfall of most clustering methods. Second, and more importantly, defects in wafer maps tend to have non-Gaussian shaped patterns (such as donut and ring). This invalidates the assumptions of many classical model-based clustering approaches that assume the clusters themselves follow a certain parametric distribution (most commonly a Gaussian). The iWMM method relaxes this assumption by making the parametric distribution assumption on the clusters in a latent space, which are related to the original complex-shaped clusters through a non-linear transformation called a warping function. Through this warping, complex non-Gaussian-like shapes in the observed space can be represented by simple Gaussian-like shapes in the latent space. Clustering is then performed in the latent space using a model-based clustering technique (e.g. Gaussian mixture). Figure \ref{fig:ex_iwmm} shows an example of how iWMM works within our framework.

\begin{figure*}
\centering
    \includegraphics[width = 1\linewidth, trim = {0 0.4cm 0 0cm}]{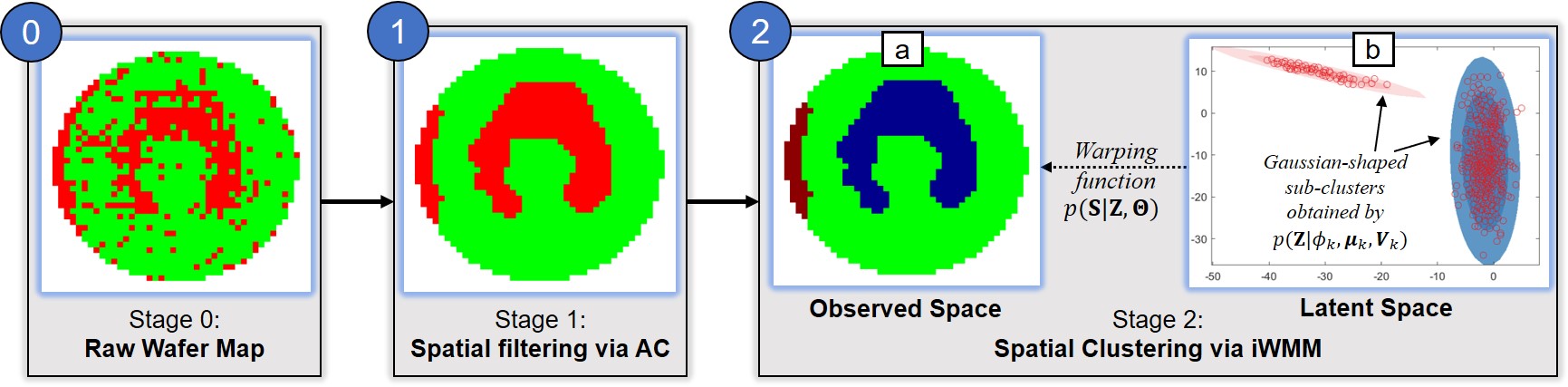}
    \caption{Our SPR consists of spatial filtering via the AC method (Stage 1) and iWMM for spatial clustering (Stage 2). iWMM assumes the non-Gaussian-shaped sub-clusters in the observed space (2a) are obtained by warping Gaussian-like sub-clusters in the latent space (2b).
    }
    \label{fig:ex_iwmm}
\end{figure*}

We briefly describe the key details of the iWMM method in our problem setting and interested readers can refer to the Appendix for more details. As evident in Figure \ref{fig:ex_iwmm}, two building blocks constitute the essence of iWMM: (1) a warping function to match the observed spatial locations of the AC-filtered results with a set of latent spatial coordinates in the latent space, and (2) a model-based clustering method which determines the clustering assignments in the latent space. The authors in \cite{iwata2012warped} propose to use a Gaussian process latent variable model (GPLVM) \cite{lawrence2004gaussian} as a warping function, and an infinite Gaussian mixture model (iGMM) as a model-based clustering method. We briefly review both in the sequel. \par 

Using the notation from Section \ref{sec:approach}, we denote by $\mathbf{S} = [\mathbf{s}_1, ..., \mathbf{s}_n]^T$ the set of spatial locations of the defective chips in the AC-filtered results, i.e. for which $x_i = 1$, where $n = \sum_{i} x_i$, and $\mathbf{s}_i \in \mathbb{R}^2$. The set $\mathbf{S}$ in the observed space corresponds to a set of latent coordinates in the latent space denoted by $\mathbf{Z} = [\mathbf{z}_1, ..., \mathbf{z}_n]^T$, where $\mathbf{z}_i \in \mathbb{R}^2$. The ultimate goal of iWMM is to find a vector of assignments in the latent space, denoted as $\mathbf{A} = [a_1, ..., a_n]^T$, where $a_i \in \mathbb{Z}^+$ denotes the membership of the $i$th chip to a particular sub-cluster. \par 

GPLVM is used to map (or warp) the transformed spatial locations $\mathbf{Z}$, which are assumed to follow a Gaussian distribution, into the observable space where $\mathbf{S}$ can have a non-Gaussian distribution. For
GPLVM, the conditional probability of $\mathbf{S}$ given $\mathbf{Z}$ is expressed as in Eq.~(\ref{eq:gplvm}) \cite{lawrence2004gaussian}. \begin{equation}
\label{eq:gplvm}
    p(\mathbf{S}|\mathbf{Z}, \pmb{\Theta}) = (2\pi)^{-n} |\pmb{\Sigma}|^{-1} \exp\bigg(-\frac{1}{2} \text{tr}(\mathbf{S}^T \pmb{\Sigma}^{-1} \mathbf{S})\bigg),
\end{equation}
where tr($\cdot$) is the trace function, $|\cdot|$ is the determinant operation, and $\pmb{\Sigma}$ is an $n \times n$ covariance matrix whose $i$th and $j$th entry holds the covariance between a pair of observations $\mathbf{z}_i$ and $\mathbf{z}_j$. To determine the entries of $\pmb{\Sigma}$, a stationary parametric covariance function $C(\cdot)$ is selected, which depends on the lag between a pair of inputs through a set of hyperparameters $\pmb{\Theta}$. A popular choice for $C(\cdot)$ is the so-called squared exponential covariance \cite{williams2006gaussian}, which is employed in this work. \par 

Once the warping function is established, iGMM is used for spatial clustering in the latent space by assuming that the $k$th mixture component (or sub-cluster) in the latent space follows a Gaussian distribution with a mean and precision matrix, denoted by $\pmb{\mu}_k$ and $\mathbf{V}_k$, respectively. Each mixture component is associated with a mixture weight, $\phi_k$. The mathematical expression for iGMM is presented in Eq.~(\ref{eq:igmm}). 
\begin{equation}
    \label{eq:igmm}
    p(\mathbf{z}|\phi_k, \pmb{\mu}_k, \pmb{V}_k) = \sum_{k=1}^{\infty} \phi_k \mathcal{N}(\mathbf{z}|\pmb{\mu}_k, \mathbf{V}_k^{-1}) 
\end{equation}

A detailed procedure to fit the iWMM to a set of observed spatial locations $\mathbf{S}$ is proposed in \cite{iwata2012warped}, where the latent coordinates $\mathbf{Z}$, assignments $\mathbf{A}$, as well as remaining parameters are inferred through a Markov Chain Monte Carlo (MCMC)-based procedure. The implementation codes for iWMM have been made publicly available \cite{codes} and we use them for our numerical analysis in Section \ref{sec:application}. 

\section{Application to Real-World Wafer Map Data} \label{sec:application}
In this section, we evaluate the performance of our proposed SPR approach on real-world wafer bin maps. 
We then derive key insights about its performance relative to a state-of-the-art approach using widely recognized clustering quality metrics.

\subsection{Data Description} 
We extract {twelve} wafer maps from a public dataset that is widely cited in the semiconductor manufacturing community \cite{data, wu2014wafer}. While the original dataset contains a large number of wafer maps, we select {twelve} wafers so as to (1) reflect different mixed-type defect patterns of various complexity, and (2) to resemble as close as possible to the six wafer maps analyzed by \cite{kim2018detection} (which we do not have access to) in order to provide a fair comparison of our proposed approach relative to the state-of-art approach in the literature. Figure \ref{fig:wafers} displays the {twelve} chosen wafer maps, where the red and green squares depict the defective and functional chips, respectively. 

\begin{figure}[h]
    \centering
        \includegraphics[width = 1\linewidth, trim = {2.25cm 3cm 2.25cm 3.5cm}]{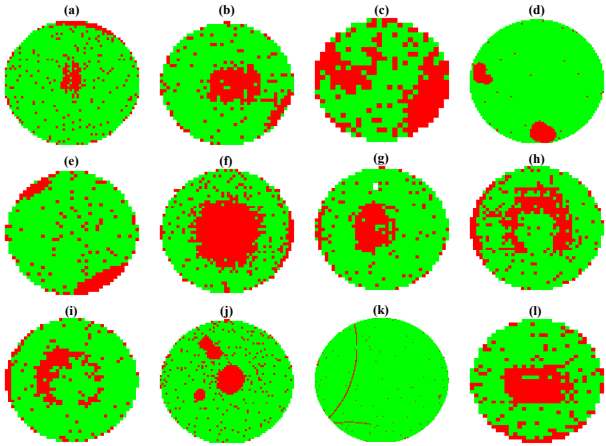}
    \caption{Twelve wafer maps with mixed-type defect patterns. (a), (f), (g), and (l): wafer maps with center and partial ring defects. (b) and (j): wafer maps with center and zone defects. (c), (d), and (e): wafer maps with two zone defects, (h) and (i): wafer maps with donut and partial ring defects, (k): a wafer map with two disconnected scratch defects.}
    \label{fig:wafers}
\end{figure}

\subsection{Results and Discussion}
Hereinafter, we denote our proposed approach as AC-iWMM where adjacency-clustering for spatial filtering is coupled with the infinite warping mixture model for spatial clustering. The benchmark in comparison is the state-of-the-art filtering approach in \cite{kim2018detection}, which is denoted hereinafter as CPF-iWMM where the connected path filtering (CPF) algorithm for spatial filtering is followed by iWMM for spatial clustering. Therefore, the fundamental difference between our approach and the benchmark lies in the spatial filtering stage, for which the impact on the quality of SPR is shown to be instrumental. \par   

We test the AC model with a king-move neighborhood system, and standardize $f_i(0,1) = f_i(1,0) = 1$ (i.e., $|w_i| = 1$) for all $i\in V$, and set $u_{ij} = u$ for all $[i,j] \in E$. The value of $u$ thus controls the spatial filtering level. In theory, we can choose the value of $u$ through a cross validation procedure as described by \cite{hochbaum2018adjacency}. As discussed in subsection \ref{subsec:ac}, having a too large or too small value of $u$ is not ideal for the spatial filtering. {We observe that for almost all defect patterns, the choice of $u=0.5$ achieves a sensible trade-off between the deviation and separation costs, and consistently yields superior performance in filtering various defect pattern combinations.}  
This is in contrast to CPF for which there does not appear to be a value for its main parameter $M$ that works universally well for different defect types (findings to be discussed in the sequel). 
We implement the CPF algorithm following the description of the method by \cite{kim2018detection}, while for iWMM, we adapt the codes available in \cite{codes}. Figure \ref{fig:result1} shows the results of the AC-iWMM approach for a subset of the wafers, starting from the raw map, to AC-filtered results, to the clustering results in both the original and latent spaces. The visual results for all wafers are shown in Figure \ref{fig:a2} in the Appendix. 

\begin{figure}[h]
    \centering
       \includegraphics[width = 1.2\linewidth, trim = {4cm 4.5cm 0cm 4cm}]{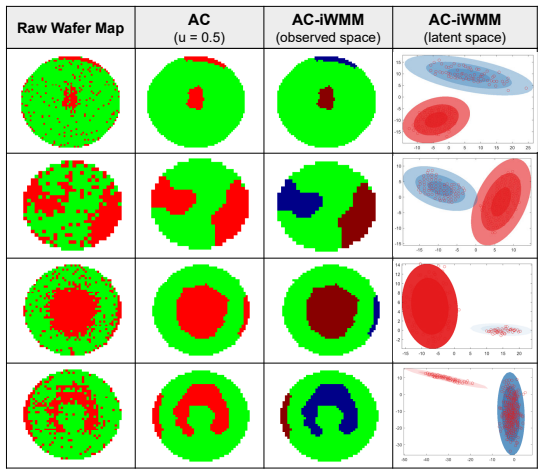}
    \caption{Results from a subset of wafers, starting from the raw maps (first column), to AC filtering using $u = 0.5$, and then to iWMM clustering, in both the original and latent spaces (third and fourth columns, respectively).}
    \label{fig:result1}
\end{figure}

\subsubsection{Visual Comparisons} 
Before we present the quantitative results, we first draw some insights based on visual comparisons between our approach and the benchmark, CPF-iWMM. Figures \ref{fig:result2} and \ref{fig:result3} show the results of both approaches on two wafer maps. The first wafer map, illustrated in Figure \ref{fig:result2}, hosts donut and partial ring defect patterns. By virtue of AC filtering, iWMM is able to distinguish the two types of defects into two separate sub-clusters that are spatially distinct. This is achieved by correctly smoothing out the random noises between the two patterns with the use of local neighborhood information. In contrast, CPF  mistakenly identifies some chips located in proximity to both sub-clusters as a separate sub-cluster, as it overlooks local neighborhood information. We note that iWMM was run at the same parameter settings for AC-iWMM and CFP-iWMM so the difference between the two sets of results is solely attributed to the spatial filtering approach.\par 

\begin{figure}[h]
    \centering
       \includegraphics[width = 1\linewidth, trim = {0cm 0.5cm 0cm 0.5cm}]{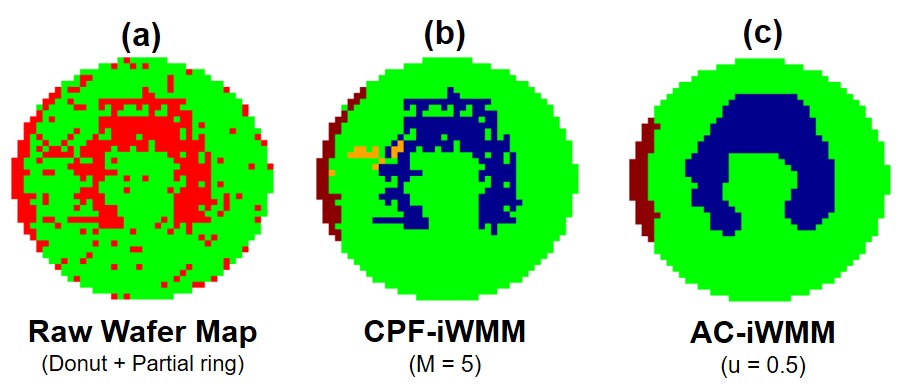}
    \caption{Visual comparison of CPF-iWMM (b) and AC-iWMM (c) on a wafer with donut and partial ring defects. Unlike AC, CPF fails to separate the two sets of defects, causing iWMM to mistakenly flag a separate sub-cluster.     }
    \label{fig:result2}
\end{figure}

Another illustrative example is shown in Figure~\ref{fig:result3}, in which the wafer map exhibits two zone defects. Again, CPF fails to separate the two sets of defective chips, causing iWMM to mistakenly flag a new separate sub-cluster. This is in contrast to AC-iWMM which yields a clear distinction between the two sub-clusters. We note that this problem cannot be alleviated by simply tuning the value of $M$ because the set of chips that are mistakenly flagged by CPF-iWMM are connected to one of the true sub-clusters, and hence, CPF will always treat it as one connected path. AC-iWMM does not keep this set of chips after AC filtering because the neighborhood information is utilized to smooth them out, reducing potential mishaps in the subsequent iWMM clustering stage. 

\begin{figure}[h]
    \centering
       \includegraphics[width = 1\linewidth, trim = {0cm 0.5cm 0cm 0.75cm}]{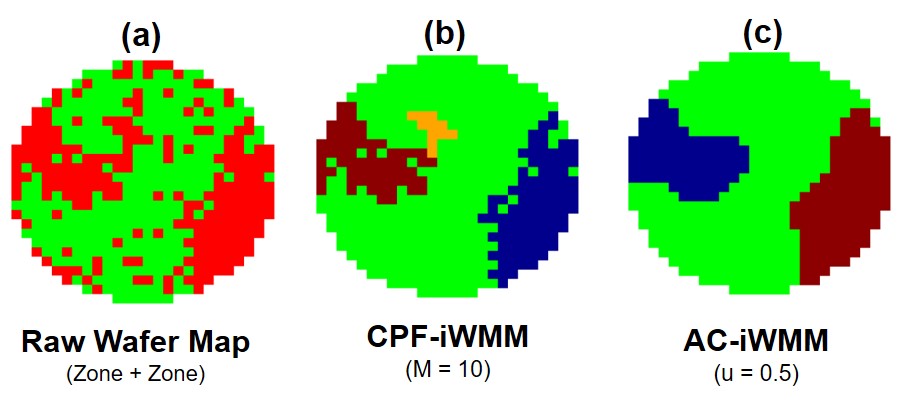}
    \caption{Visual comparison of CPF-iWMM (b) and AC-iWMM (c) on a wafer with two zone defects. In contrast to AC, CPF fails to separate the two sets of defective chips, causing iWMM to mistakenly flag a separate sub-cluster. }
    \label{fig:result3}
\end{figure}

\subsubsection{Quantitative Comparisons:} \label{sec:quant} The clustering results obtained by both approaches are then evaluated using two sets of performance metrics that are known in the SPR literature as \textit{internal} and \textit{external} indices \cite{arbelaitz2013extensive}. Assuming $\hat{\mathbf{A}} = [\hat{a}_1, ..., \hat{a}_n]^T$ and $\mathbf{A} = [a_1, ..., a_n]^T$ are the sets of predicted and true cluster assignments, respectively, internal indices assess SPR quality when the underlying ground truth is not available, that is, without access to the set $\mathbf{A}$. External indices, on the other hand, make use of $\mathbf{A}$ to validate the estimated SPR results. \par 

Let us denote by $\mathbf{G}$ the center of all coordinates in $\mathbf{S}$. Similarly, $\mathbf{G}^k$ denotes the center of the coordinates in $\mathbf{S}^k = \{\mathbf{s}_i^k\}_{i:\hat{a}_i = k}$, i.e., the coordinates of the observations assigned to the $k$th sub-cluster (for which $\hat{a}_i = k$). Two widely recognized internal indices are the Calinski-Harabasz (CH) index and the Generalized Dunn index. The Calinski-Harabasz (CH) index calculates a weighted ratio of between-cluster and within-cluster dispersion  and is defined  in Eq.~(\ref{eq:ch}) \cite{calinski1974dendrite}. By definition, a higher value for the CH index indicates a better performance. 
\begin{equation}
\small
\label{eq:ch}
    \text{CH}(\mathbf{S},\mathbf{S}^1,...,\mathbf{S}^{\hat{K}}) = \frac{n-\hat{K}}{\hat{K}-1} \frac{\sum_{k=1}^{\hat{K}} n_k ||\mathbf{G}^k - \mathbf{G}||^2}{\sum_{k=1}^{\hat{K}}\sum_{i:\hat{a}_i=k}||\mathbf{s}_i^{k}-\mathbf{G}^k||^2},
\end{equation}
where $||\cdot||$ is the Euclidean norm, and $\hat{K}$ is the predicted number of sub-clusters. \par 

The Generalized Dunn Index (GDI) defines a similar ratio, as expressed in Eq.~(\ref{eq:gdi}) \cite{bezdek1998some}. A higher value for GDI indicates a better performance.
\begin{equation}
\begin{gathered}
\text{GDI}(\mathbf{S}^1,...,\mathbf{S}^{\hat{K}}) = \\
\frac{\underset{k\neq k'}{\min} \hspace{0.5mm} \frac{1}{n_k+n_{k'}}\big(\sum_{i:\hat{a}_i=k}||\mathbf{s}_i^k-\mathbf{G}^k||+\sum_{j:\hat{a}_j=k'}||\mathbf{s}_j^{k'}-\mathbf{G}^{k'}||\big)}{\underset{k}{\max} \underset{i\neq j:\hat{a}_i=\hat{a}_j=k}{\max}\hspace{1mm} ||\mathbf{s}_i^k - \mathbf{s}_j^k||}.
\end{gathered}
\label{eq:gdi}
\end{equation}

In addition to these internal indices, we test the performance of our approach on a set of widely recognized external indices. The motivation is that, in practice, domain experts can provide the ground truth for a set of testing wafer data which can be used to assess the performance of the competing SPR approaches. Since our dataset does not have the ``ground truth,'' or in other words the set $\mathbf{A}$, we reconstruct the ground truth by applying a pattern reconstruction technique which iterates over every pixel of the raw map and updates its value using a weighted sum of its surrounding pixels to generate an output image \cite{gonzalez2002digital}. For our application, we used a $3 \times 3$ neighborhood system with $\frac{4}{9}$ weight. We also observed that this weight selection had minimal impacts on the final results. \par

Three prevalent external indices are the Rand index (RI), adjusted Rand index (ARI), and normalized mutual information (NMI). The first two metrics are based on counting pairs of observations on which the predicted clustering results agree or disagree with the true clustering assignment. Specifically, let us assume that $K$ and $\hat{K}$ are the true and predicted number of sub-clusters, respectively, and that $n_{ij}$ denote the number of observations that are common in the $i$th sub-cluster of $\mathbf{A}$ and the $j$th sub-cluster of $\hat{\mathbf{A}}$. Now, let us define $\gamma$ as the number of pairs pertaining to the same sub-cluster in $\mathbf{A}$ and to the same sub-cluster in $\hat{\mathbf{A}}$, while $\beta$, on the other hand, denotes the number of pairs pertaining to different sub-clusters in $\mathbf{A}$ and different sub-clusters in $\hat{\mathbf{A}}$. With the above notations, RI, first introduced in \cite{rand1971objective}, can be defined as:
\begin{equation}
    \text{RI}(\mathbf{A}, \hat{\mathbf{A}}) = \frac{\gamma + \beta}{\binom{2}{n}} \in [0,1],
    \label{eq:ri}
\end{equation}
where in case of perfect clustering, RI = 1, and in general, the higher its value, the better. \par 

 The second metric is the adjusted Rand index, or in short ARI, and is computed as follows:
\begin{equation}
    \small
    \label{eq:ari}
    \text{ARI}(\mathbf{A},\hat{\mathbf{A}}) =
    \frac{\binom{2}{n}(\gamma + \zeta)-[(\gamma+\beta)(\gamma+\tau)+(\tau+\zeta)(\beta+\zeta)]}{\binom{2}{n}^2 - [(\gamma+\beta)(\gamma+\tau) + (\tau+\zeta)(\beta+\zeta)]},   
\end{equation}
where $\tau$ denotes the number of pairs pertaining to the same sub-cluster in $\mathbf{A}$ and to different sub-clusters in $\hat{\mathbf{A}}$, while, $\zeta$ denotes the number of pairs pertaining to different sub-clusters in $\mathbf{A}$ and to the same sub-cluster in $\hat{\mathbf{A}}$. Similar to RI, a higher value of ARI indicates better performance. \par 

The third external metric is NMI \cite{vinh2010information}, which is an information-theoretic metric that measures the amount of information that $\mathbf{A}$ and $\hat{\mathbf{A}}$ share, and is expressed as in Eq.~(\ref{eq:nmi}). NMI ranges between $0$ and $1$, with higher values indicating better performance.
\begin{equation}
    NMI(\mathbf{A}, \hat{\mathbf{A}}) = \frac{I(\mathbf{A}, \hat{\mathbf{A}})}{H(\mathbf{A}, \hat{\mathbf{A}})} \in [0,1],
    \vspace{-.5cm}
\end{equation}
such that
\begin{equation}
    \label{eq:nmi}
    \begin{split}
        I(\mathbf{A}, \hat{\mathbf{A}}) &= \sum_{i=1}^{K} \sum_{j=1}^{\hat{K}} \frac{n_{ij}}{n} \log\bigg(\frac{n_{ij}/{n}}{(\sum_{j=1}^{\hat{K}} n_{ij}) (\sum_{i=1}^{K} n_{ij})/n^{2}}\bigg)\\
        H(\mathbf{A}, \hat{\mathbf{A}}) &= - \sum_{i=1}^{K} \sum_{j=1}^{\hat{K}} \frac{n_{ij}}{n} \log\bigg(\frac{n_{ij}/{n}}{(\sum_{i=1}^{K} n_{ij})/n}\bigg).
    \end{split}
\end{equation}

\begin{table}
     \caption{{Internal indices (CH and GDI) for all $12$ wafers. Bold-faced values indicate best performance. {*} denotes $u = 0.40$.} \vspace{-2mm}}
    \setlength{\tabcolsep}{1pt}
    \centering
    \begin{tabular}{|c|c|c|c||c|c|c|}
    \hline
    & \multicolumn{3}{c||}{CH ($\uparrow$)} & \multicolumn{3}{c|}{GDI ($\uparrow$)} \\
    \hline
    \vtop{\hbox{\strut \scriptsize Wafer}\hbox{\strut   \hspace{1.5mm} \small \#}} & \vtop{\hbox{\strut \scriptsize CPF-iWMM}\hbox{\strut   \hspace{3mm} \scriptsize (M=5)}} &\vtop{\hbox{\strut \scriptsize CPF-iWMM}\hbox{\strut   \hspace{2mm} \scriptsize (M=10)}}& \vtop{\hbox{\strut \scriptsize AC-iWMM}\hbox{\strut   \hspace{1mm} \scriptsize (u=0.50)}} & \vtop{\hbox{\strut \scriptsize CPF-iWMM}\hbox{\strut   \hspace{3mm} \scriptsize (M=5)}} &\vtop{\hbox{\strut \scriptsize CPF-iWMM}\hbox{\strut   \hspace{2mm} \scriptsize (M=10)}}& \vtop{\hbox{\strut \scriptsize AC-iWMM}\hbox{\strut   \hspace{1mm} \scriptsize (u=0.50)}} \\ 
    \hline
    1   & 539.0  &  539.0  &   \textbf{703.0}  & \textbf{.195}    &  \textbf{.195}   &   \textbf{.195} \\
    \hline
    2  & 99.0  &  99.0 &   \textbf{304.4} & .045   &   .045   &   \textbf{.281}   \\
    \hline
     3 &   517.4 &  307.3  &   \textbf{593.4}  &   \textbf{.229}    &  .176    &    {.223} \\
    \hline
    4 &   {7800.5}  &   {7800.5}  &   \textbf{8266.1}    &   \textbf{.303}  &   \textbf{.303}   &   .300   \\
    \hline
    5 &   1302.5  &   1214.6  &   \textbf{1565.8} &   .210   &  .213    &   \textbf{.215}  \\
    \hline
    6 &   85.1  &   219.3 &   \textbf{254.8}  &   .160  &   .285 &   \textbf{.302}\\
    \hline
     7 &  185.9  &   185.9  &   \textbf{222.8} &   \textbf{.276} & \textbf{.276} & \textbf{.276} \\
    \hline
   8 &  102.3 &   91.8  &   \textbf{148.4} &  .167    & .139 &   \textbf{.260} \\
    \hline
    9 &   114.2  &   200.4  &   \textbf{240.8} &   .208   & .302    &   \textbf{.320}  \\
    \hline
    10 &   455.8  &  341.9  &   \textbf{682.3} &   .130  &  .141    &   \textbf{.228}  \\
    \hline
    \hspace{1mm}11* &  38.8  &   \textbf{137.9} &  54.7   &    .011  &  \textbf{.226}    &   .184  \\
    \hline
        12 &   190.4  &  176.1  &   \textbf{218.0}  &   .236  &  .284  &   \textbf{.285}  \\
    \hline
    
    \end{tabular}
  
    \label{tab:results_int}
\end{table}

\begin{table*} [h]
    \caption{{External indices (RI, ARI, and NMI) for all $12$ wafers. Bold-faced values indicate best performance. * denotes $u = 0.40$.}}
            \vspace{-2mm}
    \setlength{\tabcolsep}{1pt}
    \centering
    \begin{tabular}{|c|c|c|c||c|c|c||c|c|c|}
    \hline
    & \multicolumn{3}{c||}{RI ($\uparrow$)} & \multicolumn{3}{c||}{ARI ($\uparrow$)} & \multicolumn{3}{c|}{NMI ($\uparrow$)}\\
    \hline
    \vtop{\hbox{\strut \footnotesize Wafer}\hbox{\strut   \hspace{1.5mm} \small \#}} & \vtop{\hbox{\strut \footnotesize CPF-iWMM}\hbox{\strut   \hspace{3mm} \footnotesize (M=5)}} &\vtop{\hbox{\strut \footnotesize CPF-iWMM}\hbox{\strut   \hspace{2mm} \footnotesize (M=10)}}& \vtop{\hbox{\strut \footnotesize AC-iWMM}\hbox{\strut   \hspace{1mm} \footnotesize (u=0.50)}} & \vtop{\hbox{\strut \footnotesize CPF-iWMM}\hbox{\strut   \hspace{3mm} \footnotesize (M=5)}} &\vtop{\hbox{\strut \footnotesize CPF-iWMM}\hbox{\strut   \hspace{2mm} \footnotesize (M=10)}}& \vtop{\hbox{\strut \footnotesize AC-iWMM}\hbox{\strut   \hspace{1mm} \footnotesize (u=0.50)}} & \vtop{\hbox{\strut \footnotesize CPF-iWMM}\hbox{\strut   \hspace{3mm} \footnotesize (M=5)}} &\vtop{\hbox{\strut \footnotesize CPF-iWMM}\hbox{\strut   \hspace{2mm} \footnotesize (M=10)}}& \vtop{\hbox{\strut \footnotesize AC-iWMM}\hbox{\strut   \hspace{1mm} \footnotesize (u=0.50)}} \\
    \hline
    1   & .969  &   .969  &   \textbf{.985}  & .938    &  .938    &   \textbf{.970} & .877 & .877  & \textbf{.916}   \\
    2  & .927  &   .927 &   \textbf{.956} & .849    &   .849   &   \textbf{.908}   &   .785    &   .785    &  \textbf{.845} \\
    3 &   .831  &  .830  &   \textbf{.856}  &   .612    &  .607    &    \textbf{.661} &    .660    &   .647    &   \textbf{.686}\\
    4 &   \textbf{.993}  &   \textbf{.993}  &   .990    &   \textbf{.986}  &   \textbf{.986}   &    .979   &  \textbf{.950}   &   \textbf{.950}   &   .936     \\
    5 &   .976  &   .974  &   \textbf{.985} &   .952    &  .948    &   \textbf{.969}   &   .897    &   .893    &   \textbf{.917} \\
    6 &   .900  &   .901 &   \textbf{.943}  &   .772    &  .776    &   \textbf{.869}   &   .751    &   .770    &   \textbf{.841}  \\
    7 &   .940  &   .940  &   \textbf{.967} &   .877    &  .877    &\textbf{.932}  &   .817    &   .817    &  \textbf{.876} \\
    8 &   .919  &   .922  &   \textbf{.939} &   .837    & .842 &   \textbf{.877}   &   .762    &   .778   &   \textbf{.812} \\
    9 &   .872  &   .871  &   \textbf{.923} &   .726    &  .724    &   \textbf{.830}   &   .699    &   .711    &   \textbf{.790} \\
    10 &   .977  &   .973  &   \textbf{.980} &   .954    &  .946    &   \textbf{.961}   &   .888    &   .880    &   \textbf{.906} \\
    \hspace{1mm}11* &  .985 & \textbf{.987}  &   .985   &   .968  &   \textbf{.971}    &   .968   &   .938    &   \textbf{.946} &   .936    \\
    12 &   .894  &   .885  &   \textbf{.931} &   .775    &  .759    &  \textbf{.851}   &   .743    &   .738    &   \textbf{.799} \\
    \hline
    
    \end{tabular}
    \label{tab:results_ext}
\end{table*}

Tables \ref{tab:results_int} and \ref{tab:results_ext} summarize the comparison results, in terms of internal and external metrics, respectively. We have included results of the CPF approach at two different values of the threshold, namely $M = 5$ and $M = 10$. With a lack of a systematic way to select $M$, those two values are selected as representatives of a low and high value, respectively. All internal and external metrics are computed using the statistical programming software $\texttt{R}$. Specifically, values of RI and ARI are computed by using functionalities in the library \texttt{fossil} \cite{vavrek2011fossil}, while NMI is computed by calling the library \texttt{NMI}. All internal indices are computed by using functionalities in the library \texttt{clusterCrit} \cite{desgraupes2013clustering}. \par 

As shown in Tables \ref{tab:results_int} and \ref{tab:results_ext}, we find that, in all wafers and across all metrics, AC-iWMM either outperforms or comes as a close second relative to CPF-iWMM with $M = 5$ or $M = 10$. We also note that the performance of the CPF approach is, in many cases, sensitive to the choice of $M$. As a case in point, varying $M$ from $5$ to $10$ in wafer \#6 changes an internal metric like CH by as much as $158$\%, and an external metric like NMI by up to $3$\%. More importantly, there is not a choice of $M$ that consistently outperforms the other. For instance, we note that a choice of $M = 5$ for wafer \#3 outperforms that of $M = 10$. In contrast, a choice of $M = 10$ for wafer \# 6 renders consistently better results than $M = 5$ across all metrics. This suggests that the choice of $M$ may be wafer-specific and requires expert judgment (as acknowledged in \cite{kim2018detection}). {As opposed to CPF, the choice of $u = 0.5$ for AC is shown to  provide consistently satisfactory performance across all defect combinations, except for the scratch patterns, where a value of $u = 0.4$ worked best. Note that, in principle, CPF is expected to perform considerably well in scratch patterns, since by design, scratch patterns are line defects, which, if longer than a carefully selected threshold (for this wafer, $M = 10$), can be naturally characterized by CPF. Nevertheless, our approach is still able to effectively distinguish the scratch defects as shown in the visual results in Figure \ref{fig:a2} (wafer \#11). We also note how changing $M$ from $10$ to $5$ for this wafer results in a substantial deterioration in performance for CPF, as shown in Tables \ref{tab:results_int} and \ref{tab:results_ext}.}  \par  

\begin{table*}[h]
    \caption{{Improvement (in percentage) of AC-iWMM over CPF-iWMM with $M = 5$, $10$, for all metrics (internal and external) across the $12$ wafers. Gray-coloured cells denote instances where percentage improvement was negative.}}
            \vspace{-2mm}
    \setlength{\tabcolsep}{1pt}
    \centering
    \begin{tabular}{|c|c|c||c|c||c|c||c|c||c|c|}
    \hline
    & \multicolumn{4}{c||}{Internal Indices} &
     \multicolumn{6}{c|}{External Indices}\\
    \hline
    & \multicolumn{2}{c||}{CH} &
    \multicolumn{2}{c||}{GDI} & 
    \multicolumn{2}{c||}{RI} & \multicolumn{2}{c||}{ARI} & \multicolumn{2}{c|}{NMI}\\
    \hline
    Wafer \# & M = 5 & M = 10 & M = 5 & M = 10 & M = 5 & M = 10 & M = 5 & M = 10 & M = 5 & M = 10 \\
    \hline
        1   & 30.4\% & 30.4\% &  0.00\% &  0.00\% &     1.65\%    &  1.65\% &  3.41\% & 3.41\% & 4.45\%   & 4.45\% \\
        
     2  & 208\% &    208\% &    524\% &    524\%  &     3.13\%  &     3.13\%  &     6.95\%  &     6.95\% &   7.64\%    &  7.64\% \\
     
    3 &   14.7\% &    93.1\%  &    \cellcolor{gray!50}{-2.62\%} &    26.7\% &    3.01\%    &  3.13\% & 8.01\%  &  8.90\% & 3.94\%  &
6.03\% \\

    4 &   5.97\% & 5.97\% & \cellcolor{gray!50}{-0.99\%}  & \cellcolor{gray!50}{-0.99\%}  & \cellcolor{gray!50}{-0.30\%}  &   \cellcolor{gray!50}{ -0.30\%} &   \cellcolor{gray!50}{ -0.71\%} &   \cellcolor{gray!50}{ -0.71\%} &   \cellcolor{gray!50}{ -1.47\%} &   \cellcolor{gray!50}{ -1.47\%} \\
   
    5 &   20.2\% &   28.9\%   &    2.38\%   &  0.94\% &    0.92\%   &     1.13\%  &   1.79\%    &     2.22\%    &     2.23\% &     2.69\% \\
    
    6 &   199\%   &   16.2\%   &     88.8\% &      5.96\%    &      4.78\% &      4.66\% &     12.6\%   &     12.0\% &     12.0\% &      9.22\%\\
    
     7 &   19.9\% &   19.9\% & \cellcolor{gray!50}{-0.32\%}  &   \cellcolor{gray!50}{ -0.32\%} &  2.87\% &    2.87\%  &   6.27\% &   6.27\% &  7.22\% &  7.22\% \\
     
        8 &   45.0\% &   61.7\%   &    55.7\%  & 87.1\%  &  2.18\%    &  1.84\% &   4.78\%    &  4.16\%  &  6.56\%  &  4.37\% \\
        
     9 &   111\% &   20.1\% &     53.9\% &      5.96\%  & 5.85\%    &   5.97\% &   14.3\%  &   14.6\% &  13.0\% &     11.1\% \\
      
      10 &   50.1\% &   100\% &     75.2\% &      61.68\%  & 0.32\%    &   0.77\% &   0.66\%  &   1.58\% &  1.95\% &     2.95\% \\
           
    11 &   41.1\% &   \cellcolor{gray!50}{-60.29\%} &     1594\% &      \cellcolor{gray!50}{-18.6\%}  & 0.00\%    &  \cellcolor{gray!50}{-0.20\%} &   0.00\%  &   \cellcolor{gray!50}{-0.42\%} &  \cellcolor{gray!50}{-0.30\%} &     \cellcolor{gray!50}{-1.10\%} \\
    
    12 &   14.5\% &   23.82\% &     20.31\% &      0.20\%  & 4.16\%   &  5.17\% &  9.77\%  &   12.11\% &  7.62\% &     8.31\%\\ 
        
    \hline
 
 Avg. & 63.3\%   &  45.6\%    &  201\%    &   58.0\%    &   2.36\%    &  2.49\%   &   5.66\%    &   5.93\%   &  5.40\%    & 5.12\%  \\
    \hline
    
    \end{tabular}
    \label{tab:results_imp}
\end{table*}

Table \ref{tab:results_imp} presents the percentage improvements of AC-iWMM relative to CPF-iWMM at $M = 5$, $10$, for all metrics. On average (last row of Table \ref{tab:results_imp}), AC-iWMM achieves an average improvement of up to {$201$}\% over CPF-iWMM in terms of internal metrics, and up to {$6$}\% in terms of external metrics. The difference in scale between the improvements in internal and external indices is attributed to how these metrics are defined in the first place; As described earlier, internal metrics are used to assess the clustering quality \textit{sans} externally provided information about the underlying cluster labels. While external validation metrics are perhaps more interpretable than their internal counterparts, the latter can be extremely useful in practice, since it may be cumbersome for experts to constantly weigh in and provide external information about class labels for all tested wafers. In other words, internal validation provides an automated check point to evaluate the method's performance in real-time. {To confirm the considerable improvement brought by AC-iWMM, we perform the Wilcoxon signed ranked test, which is a nonparametric statistical test of hypothesis, with its null hypothesis suggesting that the difference between a pair of samples follows a symmetric zero-centered distribution. The resulting $p$-values are shown in Table \ref{tab:results_sign}, wherein improvements are shown to be statistically significant for all metrics at a significance level of $0.01$, except for improvements in GDI which are significant at a $0.1$ significance level.}

\begin{table*}[h]
    \caption{{Wilcoxon signed rank test results. Each entry shows the resulting $p$-value for the corresponding metric.}}
            \vspace{-2mm}
    \setlength{\tabcolsep}{1pt}
    \centering
    \begin{tabular}{|c|c||c|c||c|c||c|c||c|c|}
    \hline
    \multicolumn{4}{|c||}{Internal Indices} &
     \multicolumn{6}{c|}{External Indices}\\
    \hline
    \multicolumn{2}{|c||}{CH} &
    \multicolumn{2}{c||}{GDI} & 
    \multicolumn{2}{c||}{RI} & \multicolumn{2}{c||}{ARI} & \multicolumn{2}{c|}{NMI}\\
    \hline
     M = 5 & M = 10 & M = 5 & M = 10 & M = 5 & M = 10 & M = 5 & M = 10 & M = 5 & M = 10 \\
    \hline
    $5 \times 10^{-4}$    &   $7 \times 10^{-3}$  &   $3 \times 10^{-2}$ &  $8 \times 10^{-2}$ & $3 \times 10^{-3}$ & $2 \times 10^{-3}$ & $3 \times 10^{-3}$ & $2 \times 10^{-3}$   &   $2 \times 10^{-3}$ & $2 \times 10^{-3}$      \\
    \hline
    
    \end{tabular}
    \label{tab:results_sign}
\end{table*}

Another interesting observation is the magnitude of improvement realized by AC-iWMM over CPF-iWMM as a function of the defect pattern complexity. Specifically, we note that improvements from AC-iWMM are more pronounced for more complex-shaped defect patterns, and are diminishing as the defect patterns become relatively simpler. For instance, substantial improvements (maximal for some metrics) in Tables \ref{tab:results_int} and \ref{tab:results_ext} come from wafer map \#9, hosting donut and partial ring defect patterns. The results for that wafer is shown in Figure \ref{fig:result6}. Understandably, donut and partial ring defect patterns are, by design, intricate shapes, making the distinction between random and systematic defects a much harder task. This is where the AC approach, through exploiting the local spatial information, can play an instrumental role in improving the quality of the spatial filtering step, and eventually, the downstream clustering and pattern recognition. On the other hand, CPF does not make use of local neighborhood information, which causes the downstream clustering to misidentify several defective chips as an independent sub-cluster.
\begin{figure} [h]
    \centering
    \includegraphics[trim = 0 0.5cm 0 0.5cm, width = 1 \linewidth]{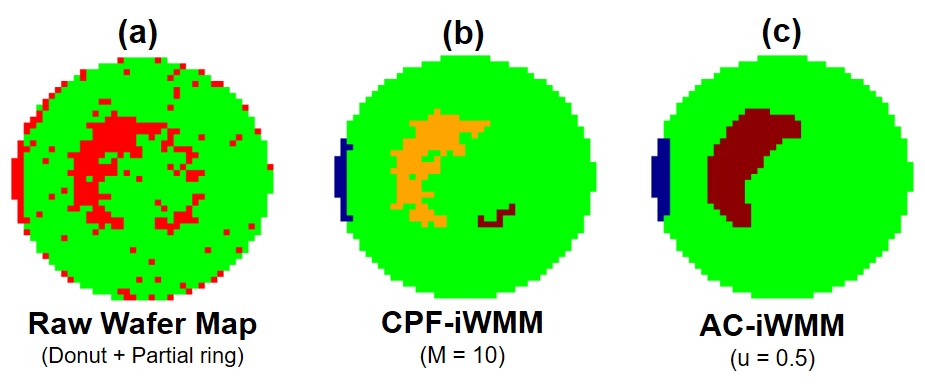}
    \caption{Visual comparison of AC-iWMM and CPF-iWMM on wafer map \#9 with donut and partial ring defects. We note that CPF fails to separate the two sets of defective chips, causing iWMM to mistakenly flag a separate sub-cluster. Here, AC-iWMM achieves substantial improvements over CPF-iWMM owing to its ability to better filter complex-shaped defect patterns.
    }
    \label{fig:result6}
\end{figure}

\begin{figure} [h]
    \centering
    \includegraphics[trim = 0 0.5cm 0 0.5cm,width = 1 \linewidth]{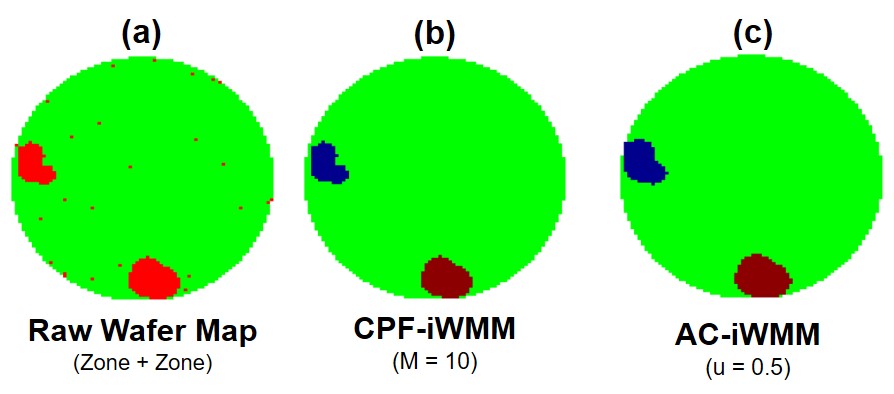}
    \caption{Visual comparison of AC-iWMM and CPF-iWMM on wafer \#4 with two zone defects. We note that both approaches render similar results, visually, and quantitatively. The marginal difference is due to the simplicity of the defect patterns\textemdash two zone defects with sparse random noises in the background, which are effectively filtered by both methods.}
    \label{fig:result5}
\end{figure}

In contrast, wafer map \#4 has a relatively simple mixed-type defect pattern, in which the two zone defects are round shaped and far from each other. Furthermore, the random defects outside the two zones are relatively sparse, which makes the filtering task straightforward. Therefore, both methods were able to produce satisfactory performance, with almost negligible visual differences, as shown in Figure \ref{fig:result5}.

The observations from Figures \ref{fig:result6} (for wafer map \#9) and \ref{fig:result5} (for wafer map \#4) validate our conjecture that the difference in performance of AC-iWMM relative to CPF-iWMM hinges on the complexity of the underlying defect patterns. A closer look at the results in Table \ref{tab:results_imp} suggests the same observation for wafer maps \#6 (visual result shown in Figure \ref{fig:a2} in the Appendix) and \#8 (visual result shown in Figure \ref{fig:result2}), which have complex-shaped patterns, and hence, the benefit from AC-iWMM appears to be more pronounced. As the wafer fabrication process grows in scale and sophistication, owing to technology upgrades, or an increase in the number of processing steps or the density of chips per wafer, wafers are expected to exhibit more intricate and mixed-type defect patterns. Thus, we anticipate that our proposed approach will generate even more profound impacts. \par 

\section{Conclusions}\label{sec:conclus}
In this paper we have proposed a spatial pattern recognition framework (AC-iWMM) for detecting  mixed-type defect patterns in wafer bin map data\textemdash a problem of vital importance to ensuring quality control in the semiconductor manufacturing industry. This framework integrates the adjacency-clustering (AC) model for spatial filtering with an advanced mixture model (iWMM) for spatial clustering. AC has a desirable combinatorial structure and can be solved in polynomial time by a minimum-cut algorithm. 
By utilizing the local neighborhood information, AC is able to effectively distinguish the systematic patterns from random noises. As a result, iWMM, which subsequently acts on the AC-filtered data, can properly cluster the systematic patterns into different types. We validate the superior performance of AC-iWMM on {twelve} real-world wafer bin maps exhibiting different mixed-type defect patterns. Based on both visual and quantitative comparisons, AC-iWMM outperforms the state-of-the-art method in the literature, especially for complex-shaped, mixed-type patterns.

The framework proposed herein can be extended in various directions. One interesting idea is to incorporate additional features to aid with mixed-type SPR such as the number of defects per bin on the premise that defects due to the same assignable cause may exhibit similar defect severity levels. Incorporating this information will drive a departure from using binary wafer maps, which, by consequence, will require a new definition of what constitutes a cluster in our setting.

\appendices
\section{Additional details about iWMM}
Here we provide additional details about iWMM, which was initially proposed by \cite{iwata2012warped}. 
iWMM comprises two building blocks:  (1) a warping function to match the observed spatial locations of the AC-filtered results, denoted by $\mathbf{S} = [\mathbf{s}_1, ..., \mathbf{s}_n]^T$, with a set of latent spatial coordinates in a latent space, denoted by $\mathbf{Z} = [\mathbf{z}_1, ..., \mathbf{z}_n]^T$, and (2) a clustering method which determines the clustering assignments in in the latent space, denoted by $\mathbf{A} = [a_1, ..., a_n]^T$. While in theory, $\mathbf{z}_i$ can have a different dimensionality than $\mathbf{s}_i$, it suffices in our setting to assume that both $\mathbf{s}_i, \mathbf{z}_i \in \mathbb{R}^2$.\par 

As a warping function, a Gaussian process latent variable model (GPLVM) \cite{lawrence2004gaussian} with squared exponential covariance, is used, and can be expressed as in Eq.~(\ref{eq:gplvm}). Then, the infinite Gaussian mixture model (iGMM) is used for spatial clustering, and is expressed as in Eq.~(\ref{eq:igmm}). For iGMM, a Gaussian-Wishart prior is placed on its parameters $\pmb{\mu}_k$ and $\mathbf{V}_k$, such that: 
\begin{equation}
    p(\pmb{\mu}_k,\mathbf{V}_k) = \mathcal{N}(\pmb{\mu}_k|\mathbf{m}, (p\mathbf{V}_k)^{-1})\mathcal{W}(\mathbf{V}_k|\mathbf{R}^{-1}, r),
\end{equation}
where $\mathcal{W}(\cdot)$ is the Wishart distribution. The parameters $\mathbf{m}$, $p$ are the mean and relative precision of $\pmb{\mu}_k$, respectively, while $\mathbf{R}^{-1}$ and $r$ are the scale matrix for $\mathbf{V}_k$, and its degree of freedom, respectively. One can then derive the probability distribution of $\mathbf{Z}$ given the clustering assignments $\mathbf{A}$ by integrating out $\pmb{\mu}_k$ and $\mathbf{V}_k$, as in Eq.~(\ref{eq:Z}). 
\begin{equation}
    p(\mathbf{Z}|\mathbf{A},\mathbf{R},\mathbf{m},r,p) = \prod_{k=1}^{\infty} \pi^{-n_k} \frac{p|\mathbf{R}|^{r/2}}{p_k|\mathbf{R}_k|^{r_{k}/2}} \times \prod_{j=1}^{2} \frac{\Gamma(\frac{r_k+1-j}{2})}{\Gamma(\frac{r+1-j}{2})},
    \label{eq:Z}
\end{equation}
where $n_k$ is the number of chips assigned to the $k$th sub-cluster, while $p_k$, $r_k$, and $\mathbf{S}_k$ are the posterior Gaussian-Wishart parameters of the $k$th component (or sub-cluster), such that $p_k = p + n_k$, $r_k = r + n_k$, and $\mathbf{S}_k = \mathbf{S} + \sum_{i:a_i=k} \mathbf{z}_i \mathbf{z}_i^T + p\mathbf{m}\mathbf{m}^T - p_k \mathbf{m}_k \mathbf{m}_k^T$, with $\mathbf{m}_k = \frac{p\mathbf{m} + \sum_{i:a_i=k} \mathbf{z}_i}{p + n_k}$. \par 

A Dirichlet process prior with concentration parameter $\alpha$ is used for infinite mixture modeling in the latent space. Then, the probability distribution of $\mathbf{A}$ can be written as:
\begin{equation}
p(\mathbf{A}|\alpha) = \frac{\alpha^k \prod_{k=1}^K (n_k - 1)!}{\alpha (\alpha+1)...(\alpha+n-1)},
\label{eq:multi}
\end{equation}
 
 Collecting the above pieces, the joint distribution of $\mathbf{S}$, $\mathbf{Z}$, and $\mathbf{A}$ conditional on all parameters, can be written as: 
\begin{equation}
\label{eq:all}
\small
    p(\mathbf{S},\mathbf{Z},\mathbf{A}|\pmb{\Theta},\mathbf{R},\mathbf{m},r, p,\alpha) = p(\mathbf{S}|\mathbf{Z},\pmb{\Theta}) p(\mathbf{Z}|\mathbf{A},\mathbf{R},\mathbf{m},r,p) p(\mathbf{A}|\alpha),
\end{equation}
 which is merely the product of the terms determined by Eqs.~(\ref{eq:gplvm}), (\ref{eq:Z}), and (\ref{eq:multi}). \par 

The authors in \cite{iwata2012warped} provide a detailed procedure to fit the iWMM to a set of observed spatial locations $\mathbf{S}$, where the latent coordinates $\mathbf{Z}$, assignments $\mathbf{A}$, as well as remaining parameters are inferred through a Markov Chain Monte Carlo (MCMC)-based procedure. The procedure consists of two steps, which are repeatedly performed until convergence. The first step entails a Gibbs sampling scheme of the latent assignment of the $i$th chip, denoted by $a_i$, from the following probability distribution: 
\begin{equation}
\label{eq:assigns}
\begin{gathered}
    p(a_i = k|\mathbf{Z},\mathbf{A}^{-i},\mathbf{R},\mathbf{m},r,p,\alpha) \\ \propto 
    \begin{cases}
    n_k^{-i} p(\mathbf{z}_i|\mathbf{Z}_k^{-i},\mathbf{R},\mathbf{m},r,p) \hspace{4mm} \text{\small assign to an existing sub-cluster}\\
    \alpha p(\mathbf{z}_i|\mathbf{R},\mathbf{m},r,p) \hspace{15mm} \text{\small form a new sub-cluster,}
    \end{cases}
    \end{gathered}
\end{equation}
where $\mathbf{Z}_k^{-i}$ is the set of latent coordinates of the $k$th sub-cluster, excluding the $i$th chip. Similarly, $\mathbf{A}^{-i}$ is the set of assignments, excluding that of the $i$th chip, and $n_k^{-i}$ is the number of chips assigned to the $k$th sub-cluster, excluding the $i$th chip. The probability distributions in the right hand-side of Eq.~(\ref{eq:assigns}) can be analytically derived in closed-form as detailed in \cite{iwata2012warped}. The second step entails sampling the latent coordinates $\mathbf{Z}$ from the probability distribution $p(\mathbf{Z}|\mathbf{A},\mathbf{S},\pmb{\Theta},\mathbf{R},\mathbf{m},r,p)$ using hybrid Monte Carlo. Combined, the two steps yield an estimate of the posterior distribution of the latent coordinates $\mathbf{Z}$ and the latent assignments $\mathbf{A}$. 
\vspace{-0.1cm}
\section{Clustering results for all wafers}
In Figure \ref{fig:a2}, we show the visual clustering results for all {twelve} wafers depicted in Figure \ref{fig:wafers}. 

\begin{figure} [!h]
    \centering
    \includegraphics[width = 1 \linewidth]{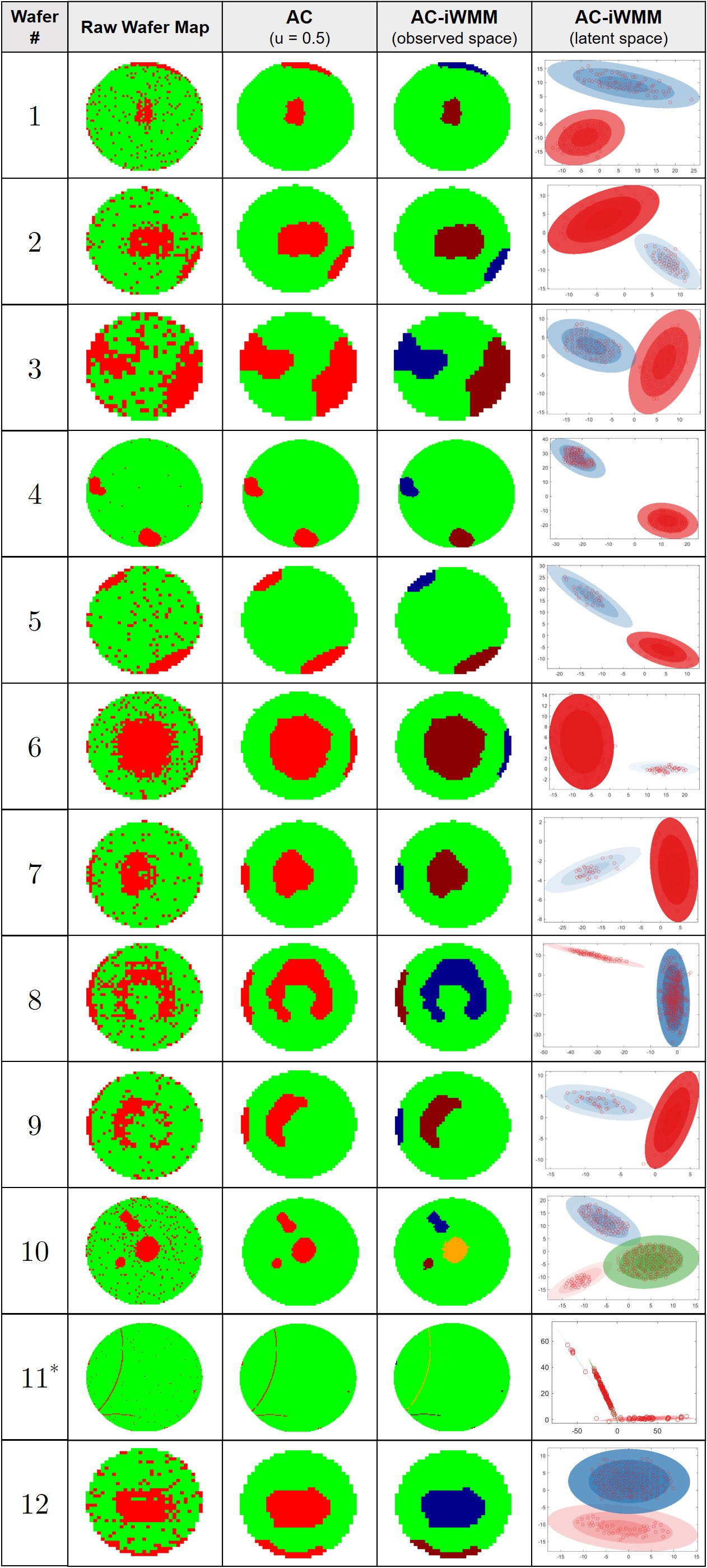}
    \caption{{Results from all $12$ wafers, starting from wafer id (first column), raw wafer maps (second column), to AC filtering results (third column), and then to iWMM clustering as applied to the AC-filtered data, in both the original and latent spaces (fourth and fifth columns, respectively). * denotes $u = 0.4$}.}
    \vspace{-0.25cm}
    \label{fig:a2}
\end{figure}
\vspace{-0.1cm}
{\section{}}
\vspace{-0.1cm}
\nomenclature{$V$}{The set of nodes (chip indices)}
\nomenclature{$E$}{The set of edges (pairs of adjacent chips)}
\nomenclature{$G$}{The graph representation of a wafer}
\nomenclature{$d_i$}{The number of defects on the $i$th chip}
\nomenclature{$x_i$}{Cluster membership for the $i$th chip}
\nomenclature{$f_i(\cdot)$}{Deviation cost function in the AC formulation}
\nomenclature{$g_{ij}(\cdot)$}{Separation cost function in the AC formulation}
\nomenclature{$z_{ij}$}{Difference in cluster labels of the $i$th and $j$th chips}
\nomenclature{$w_{i}$}{Deviation cost of the $i$th chip}
\nomenclature{$u_{ij}$}{Separation cost of the $i$th and $j$th chips (for $[i,j] \in E$)}
\nomenclature{$\mathbf{s}_i$}{The observed coordinates of the $i$th defective chip}
\nomenclature{$\mathbf{S}$}{The set of observed coordinates of the defective chips}
\nomenclature{$n$}{The number of defective chips}
\nomenclature{$\mathbf{z}_i$}{The latent coordinates of the $i$th defective chip}
\nomenclature{$\mathbf{Z}$}{The set of latent coordinates of the defective chips}
\nomenclature{${a}_i$}{The true assignment of the $i$th defective chip}
\nomenclature{$\mathbf{A}$}{The vector of true assignments of defective chips}
\nomenclature{$\hat{a}_i$}{The predicted assignment of the $i$th defective chip}
\nomenclature{$\hat{\mathbf{A}}$}{The vector of predicted assignments of defective chips}
\nomenclature{$\pmb{\Sigma}$}{The $n \times n$ covariance matrix whose $i$th and $j$th entry is the covariance between latent coordinates $\mathbf{z}_i$ and $\mathbf{z}_j$}
\nomenclature{$C(\cdot)$}{A stationary parametric covariance function}
\nomenclature{$\pmb{\Theta}$}{Hyperparameters of the covariance function $C(\cdot)$}
\nomenclature{$\pmb{\mu}_k$}{The mean of the $k$th mixture component}
\nomenclature{$\mathbf{V}_k$}{The precision matrix of the $k$th mixture component}
\nomenclature{$\phi_k$}{The weight of the $k$th mixture component}
\nomenclature{$\mathbf{G}$}{The center of all coordinates in $\mathbf{S}$}
\nomenclature{$\mathbf{S}^k$}{The chip coordinates of the $k$th sub-cluster}
\nomenclature{$\mathbf{s}_i^k$}{The $i$th chip's coordinates of the $k$th sub-cluster}
\nomenclature{$\mathbf{G}^k$}{The center of all coordinates in $\mathbf{S}^k$}
\nomenclature{$\gamma$}{The number of pairs pertaining to identical sub-clusters in $\mathbf{A}$ and $\hat{\mathbf{A}}$}
\nomenclature{$\beta$}{The number of pairs pertaining to different sub-clusters in $\mathbf{A}$ and $\hat{\mathbf{A}}$}
\nomenclature{$\tau$}{The number of pairs pertaining to identical sub-clusters in $\mathbf{A}$ and different sub-clusters in $\hat{\mathbf{A}}$}
\nomenclature{$\zeta$}{The number of pairs pertaining to different sub-clusters in $\mathbf{A}$ and identical sub-clusters in $\hat{\mathbf{A}}$}
\nomenclature{$M$}{The pre-set threshold for the CPF method}
\nomenclature{$\mathcal{W}(\cdot)$}{The Wishart distribution}
\nomenclature{$\mathbf{m}$}{The mean parameter for $\pmb{\mu}_k$}
\nomenclature{$p$}{The relative precision parameter for $\pmb{\mu}_k$}
\nomenclature{$\mathbf{R}^{-1}$}{The scale matrix for $\mathbf{V}_k$}
\nomenclature{$r$}{The number of degrees of freedom for $\mathbf{V}_k$}
\nomenclature{$\alpha$}{Concentration parameter for the Dirichlet process}
\nomenclature{$\mathbf{Z}_k^{-i}$}{The set of latent coordinates of the $k$th sub-cluster, excluding the $i$th chip}
\nomenclature{$\mathbf{A}^{-i}$}{The set of assignments, excluding the $i$th chip}
\nomenclature{$n_k$}{The number of chips assigned to the $k$th sub-cluster}
\nomenclature{$n_k^{-i}$}{The number of chips assigned to the $k$th sub-cluster, excluding the $i$th chip}
\printnomenclature

\section*{Acknowledgment}
The research of Dorit S. Hochbaum is supported in part by NSF award No. CMMI-1760102.

\ifCLASSOPTIONcaptionsoff
  \newpage
\fi



\bibliographystyle{IEEEtran}
\bibliography{references}
\end{document}